\documentclass[envcountsect]{SNmult}

\usepackage{indentfirst}

\usepackage{amsmath}
\usepackage{amssymb}
\usepackage{stmaryrd}
\usepackage{bm}
\usepackage{nicefrac}
\usepackage{newtxtext}
\usepackage[varvw]{newtxmath}
\usepackage{type1cm}
\usepackage{tikz}
\usetikzlibrary{arrows.meta}
\usepackage{hyperref}
\usepackage[bottom]{footmisc}

\newcommand{\opsize}{0.68ex}
\newcommand{\oplinewidth}{0.50pt}
\newcommand{\opsizeScript}{0.54ex}
\newcommand{\oplinewidthScript}{0.42pt}
\newcommand{\opsizeScriptScript}{0.44ex}
\newcommand{\oplinewidthScriptScript}{0.35pt}

\newtheorem{thm}{Theorem}[section]
\newtheorem{lemm}[thm]{Lemma}
\newtheorem{prop}[thm]{Proposition}
\newtheorem{defi}[thm]{Definition}

\newcommand{\modalopdrawdir}[4]{%
  \vcenter{\hbox{%
    \tikz[x=1ex,y=1ex,baseline=(current bounding box.center)]{%
      \def\r{#3}%
      \def\lw{#4}%
      \ifx#1c%
        \draw[line width=\lw] (0,0) circle [radius=\r];%
      \else\ifx#1s%
        \draw[line width=\lw] (-\r,-\r) rectangle (\r,\r);%
      \fi\fi%
      \ifx#2r%
        \draw[line width=\lw] (0,0) -- (\r,0);%
      \else\ifx#2l%
        \draw[line width=\lw] (0,0) -- (-\r,0);%
      \fi\fi%
    }%
  }}%
}

\newcommand{\ModalForward}[1]{%
  \mathord{\mathchoice
    {\modalopdrawdir{#1}{r}{\opsize}{\oplinewidth}}%
    {\modalopdrawdir{#1}{r}{\opsize}{\oplinewidth}}%
    {\modalopdrawdir{#1}{r}{\opsizeScript}{\oplinewidthScript}}%
    {\modalopdrawdir{#1}{r}{\opsizeScriptScript}{\oplinewidthScriptScript}}%
  }%
}

\newcommand{\ModalBackward}[1]{%
  \mathord{\mathchoice
    {\modalopdrawdir{#1}{l}{\opsize}{\oplinewidth}}%
    {\modalopdrawdir{#1}{l}{\opsize}{\oplinewidth}}%
    {\modalopdrawdir{#1}{l}{\opsizeScript}{\oplinewidthScript}}%
    {\modalopdrawdir{#1}{l}{\opsizeScriptScript}{\oplinewidthScriptScript}}%
  }%
}

\DeclareRobustCommand{\FCircleOp}{\ModalForward{c}}
\DeclareRobustCommand{\BCircleOp}{\ModalBackward{c}}

\newcommand{\Prop}{\mathbb{P}}
\newcommand{\peq}{\preccurlyeq}
\newcommand{\reach}[2]{\gamma(#1,#2)}
\newcommand{\lb}{\left\llbracket}
\newcommand{\rb}{\right\rrbracket}
\newcommand{\val}[1]{\lb #1 \rb}
\newcommand{\normal}{\textsf}
\newcommand{\lgc}[1]{{\mathsf{#1}}}
\newcommand{\nc}{\Box}
\newcommand{\Poly}[1]{\mathsf{PL}(#1)}
\newcommand{\IntOp}{\mathcal{I}}
\newcommand{\ClOp}{\mathcal{C}}
\newcommand{\dyn}[1]{\mathsf{hD\mbox{-}}#1}

\makeatletter
\@ifundefined{question}{\newtheorem{question}{Question}}{}
\makeatother

\begin{document}

\title*{Dynamic Polyhedral Logic}

\author{Nick Bezhanishvili\orcidID{0009-0005-6692-5051} and\\ Laura Bussi\orcidID{0000-0003-1292-4086}
and\\ Vincenzo Ciancia\orcidID{0000-0003-1314-0574} and\\ David Fern\'andez-Duque\orcidID{0000-0001-8604-4183} and\\
David Gabelaia\orcidID{0000-0002-8317-7949}}

\institute{Nick Bezhanishvili \at University of Amsterdam, Amsterdam, Netherlands, \email{n.bezhanishvili@uva.nl}
\and Laura Bussi \at University of Amsterdam, Amsterdam, Netherlands \email{l.bussi@uva.nl}
\and Vincenzo Ciancia \at National Research Council, Pisa, Italy \email{vincenzo.ciancia@isti.cnr.it}
\and David Fern\'andez-Duque \at University of Barcelona, Barcelona, Spain \email{fernandez-duque@ub.edu}
\and David Gabelaia \at TSU Razmadze Mathematical Institute, Tbilisi, Georgia \email{gabelaia@gmail.com}}

\maketitle

\abstract{%
We introduce spatio-temporal polyhedral reachability logics, extending dynamic
topological logic with polyhedral semantics and a path-based spatial reachability
operator. Formulas are interpreted over polyhedra, with admissible valuations
ranging over polyhedral subsets; the spatial modality is interpreted as interior,
the binary operator $\gamma(\varphi,\psi)$ expresses reachability of a
$\psi$-point through a $\varphi$-region, and the temporal modalities are
interpreted by a PL-homeomorphism and its inverse. We define h-dynamic
reachability spaces and axiomatize the corresponding h-dynamic extensions of
the known reachability logics of topological, finite, Alexandroff, and polyhedral
spaces. The main result is soundness and completeness for the intended classes of invertible dynamical systems. 
We are happy to contribute this paper to the volume dedicated to Sergei Artemov---one of the pioneers of combining spatial and temporal reasoning in the framework of dynamic topological logic. 
}

\section{Introduction}

Sergei Artemov's work has uncovered many deep and unexpected connections between logic,
mathematics, computer science, epistemology, and the formal analysis of information.
The present paper is situated in one such line of interaction: the use of modal logic
to reason about spaces, and processes evolving in time.

The starting point is the classical topological semantics of modal logic, going back
to McKinsey and Tarski~\cite{McKinseyTarski1944-the-algebra-of-topology}. In this
semantics, formulas are interpreted over a topological space $X$, with the modal
operator $\nc$ read as topological interior. Thus $\nc\varphi$ holds at a point when
$\varphi$ holds throughout some neighbourhood of that point. The resulting logic is
$\lgc{S4}$, also for a single space that is sufficiently ``rich'', including any dense-in-itself metric space.

Topological semantics can be combined with linear temporal logic ~\cite{G92} to produce a framework for reasoning about {\em dynamical systems,} defined as a topological space $X$ equipped with a continuous self-map $f\colon X\to X$~\cite{akin}.
We may use the topology to interpret the spatial modality and the map $f$ to interpret
the temporal `next' operator. The end result is the logic $\sf S4C$, introduced by Artemov, Davoren and Nerode~\cite{ArtemovDavorenNerode1997-modal-logics-and-topological-semantics-for-hybrid-systems}, known to enjoy several nice properties such as finite axiomatisability and the FMP. 

In the case where $f$ is a homeomorphism, the dynamics become {\em invertible,} as (by definition), $f^{-1}$ is a continuous function as well.
$\sf S4H$, the variant of $\sf S4C$ interpreted over invertible systems, enjoys similar properties, as shown by Kremer and Mints~\cite{KremerMints2005-DTL} who re-dubbed them {\em dynamic topological logics,} or DTLs.
The expectation was for such logics to be instrumental in automated theorem-proving, but with current advances in robotics and artificial intelligence, interest in spatio-temporal reasoning has spread far beyond pure mathematics (see e.g. ~\cite{selfdrive,selfdrive2}).

In the present paper we focus on this
invertible setting and use two temporal operators: $\FCircleOp$ for the forward
iteration of $f$, and $\BCircleOp$ for the backward iteration of $f^{-1}$. Combining these operators with the topologically interpreted $\nc$ allows us to reason about topological dynamics, as for example the continuity of $f$ is equivalent to the validity of $\FCircleOp \nc p\to \nc \FCircleOp p $.
Likewise, if $f$ is a homeomorphism, then $f^{-1}$ is continuous as well, leading to the validity of $\BCircleOp \nc p\to \nc \BCircleOp p $.
This provides an alternative representation of DTL with homeomorphisms, which is usually only presented with $\FCircleOp$ and with the stronger axiom $\FCircleOp \nc p\leftrightarrow \nc \FCircleOp p $.  
The latter is derivable in our system, which we will not explicitly show, but the reasoning is along the lines of the closely related Lemma~\ref{lemHomAx}.

Note however, that ordinary topological semantics with just $\nc$ interpreted as interior is often too coarse to capture genuinely
geometric information. The classical McKinsey--Tarski phenomenon shows that many
topologically different spaces validate the same modal formulas. For example, any class of dense-in-itself metric spaces has the same basic modal logic. A way to
obtain a more geometry-sensitive semantics is to restrict attention to a well-behaved algebra of admissible regions, rather than the full powerset algebra of the underlying space. This is
the guiding idea behind polyhedral semantics \cite{BezhanishviliMarraMcNeillPedrini2018,AdamDayBezhanishviliGabelaiaMarra2024,gabelaia2018modallogicplanarpolygons}.

In polyhedral semantics, the underlying spaces are polyhedra, and formulas are
interpreted as polyhedral subsets \cite{AdamDayBezhanishviliGabelaiaMarra2024}. Equivalently, one works with regions that can be
described using finite simplicial decompositions \cite{RourkeSanderson1972}. This gives a natural semantics for
piecewise-linear geometry: points, segments, triangles, tetrahedra, and their
higher-dimensional analogues. The modal logics arising from this semantics is actively being
studied; in particular, the logic of all polyhedra is 
Grzegorczyk's modal logic $\lgc{Grz}$ and is closely connected with the finite-poset combinatorics of triangulations and nerves~\cite{BezhanishviliMarraMcNeillPedrini2018}. One of the important lessons of
this line of work is that a polyhedron can often be analysed logically through the
face posets of sufficiently fine triangulations \cite{AdamDayBezhanishviliGabelaiaMarra2024}.

For applications and for geometry itself, it is natural to enrich the basic spatial
modal language by a path-based reachability operator~\cite{CLLM16}. We use a binary operator
$\gamma(\varphi,\psi)$, read informally as: ``$\psi$ is reachable through $\varphi$''.
At a point $x$ of a topological space, $\gamma(\varphi,\psi)$ holds if there is a
continuous path starting at $x$, ending at a point satisfying $\psi$, and whose
intermediate points satisfy $\varphi$. Thus $\gamma$ is a spatial analogue of an
Until operator. It can express properties that are invisible to the basic closure/interior
language, such as the existence of a safe path to an exit, propagation through an
allowed region, or the possibility of moving from one spatial component to another
without crossing a forbidden area.

This reachability operator is especially natural in polyhedral models. Continuous
paths can be replaced, for logical purposes, by piecewise-linear paths, and the
interaction between paths and triangulations makes the semantics finite and
combinatorial in a strong sense. This is also what makes the language relevant to
spatial model checking. In such
settings, formulas involving $\gamma$ can specify connectivity, containment,
surroundedness, safe reachability, and related spatial properties~\cite{CLLM16,BCGGLM22}.
When combined with temporal operators, these formulas describe spatial properties of
systems that evolve over time, a situation that arises naturally in domains such as
medical imaging, video analysis, robotics, and geometric model checking
\cite{BCLM19,BCGLM22,FASE26,BCM25}. In the context of medical imaging, anatomical structures are often clearly spatially described by protocols and guidelines. Recent work has shown that spatial model checking can be effectively employed for the declarative analysis and segmentation of magnetic resonance and radiological images~\cite{BCLM19,FASE26}, and in hybrid-AI setups in conjunction with neural networks ~\cite{BCM25}. This permits complex image analysis pipelines to be expressed as logical formulas and automatically verified on image data. Moreover, the adoption of polyhedral semantics enables reasoning not only on voxel-based images but also on geometric representations such as meshes and anatomical surfaces, which are increasingly common in advanced 3D imaging pipelines. These approaches combine explainability, formal guarantees, and computational efficiency, making spatial and spatio-temporal logics a promising foundation for next-generation medical image analysis and computer-aided diagnosis. 
It is also relevant in this context that classical Hennessy-Milner-style theorems proving correctness and completeness of bisimilarity relations with respect to their logical counterparts find fertile grounds in this extended application scenario. For instance \cite{BBCJLMV26} establishes a minimization result of practical impact (compression of images amenable to faster spatial-logical analysis, and supporting tools) stemming from such an adequacy result. 

The aim of this paper is to bring together these three ingredients: polyhedral
semantics, spatial reachability, and invertible dynamics. We introduce and study
h-dynamic polyhedral reachability models, where the underlying space is a polyhedron,
the admissible regions are polyhedral sets, the spatial operator $\nc$ is interpreted
as interior, the reachability operator $\gamma$ is interpreted by (piecewise-linear) paths, and the temporal
operators $\FCircleOp$ and $\BCircleOp$ are interpreted by a PL-homeomorphism and
its inverse.

Our main technical contribution is an axiomatization and completeness theorem for
the resulting logics. We begin with known reachability logics for topological,
Alexandroff, finite, and polyhedral spaces, and then add temporal axioms governing
the two directions of an invertible dynamics and their interaction with reachability.
The key observation is that temporal operators can be pushed down to the propositional
level. More precisely, every formula of the full spatio-temporal language can be
translated into an equivalent ``simple'' formula in the spatial reachability language,
where temporal shifts appear only on propositional variables. Completeness of the
spatio-temporal systems is then reduced to the already established completeness of
their spatial fragments.

In the polyhedral case, this reduction must also respect the geometry. Given a
polyhedral model for the translated formula, we reconstruct an h-dynamic polyhedral
model by arranging finitely many copies of the original polyhedron around a periodic
orbit. In fact, the construction may be chosen so that the dynamics
is a rotation. This yields completeness not only for arbitrary PL-homeomorphisms, but
also for a particularly transparent class of h-dynamic polyhedral systems. 
A precursor to this approach is \cite{Kopnev}, which considers a variant of dynamic polyhedral logic without the reachability operator. In fact, \cite[Section 8]{Kopnev} leaves it as an open problem to develop a dynamic polyhedral logic with the reachability operator. In this respect, the present paper addresses that problem
by providing the appropriate semantics and characterizing its logic.

The paper is organized as follows. In Section~2 we introduce the spatio-temporal
language and the general algebraic semantics for reachability models and h-dynamic
reachability spaces. Section~3 recalls the relevant polyhedral background: simplicial
complexes, open cells, polyhedral sets, PL maps, and the closure properties needed
for interpreting the language. In Section~4 we recall the axiomatic systems for the
spatial reachability logics and define their h-dynamic extensions. We prove soundness
of the temporal and reachability interaction axioms. Section~5 contains the main
completeness argument: the translation to simple formulas, the reconstruction of
h-dynamic models, and the resulting completeness theorems for topological, finite,
Alexandroff, and polyhedral semantics. Finally, Section~6 discusses open problems and
directions for further work, including non-invertible dynamics and extensions with
operators like ``eventually'' and ``henceforth''.

\section{Language and Algebraic Semantics}

In this section, we introduce our formal language and the general algebraic semantics for reachability logics.
We assume familiarity with topological spaces, and follow the convention of notationally identifying a topological space $(X,\mathcal T)$ with $X$.

Let $\Prop$ be a countably infinite set of propositional variables.

\begin{defi}[The Spatio-Temporal Language]
The full language $\mathcal{L}$ is generated by the grammar:
\[
\varphi ::= p \mid \lnot \varphi \mid \varphi \land \psi \mid \nc \varphi \mid \gamma(\varphi, \psi) \mid \FCircleOp \varphi \mid \BCircleOp \varphi \qquad (p \in \Prop)
\]
\end{defi}

The dual spatial diamond operator is defined as $\Diamond \varphi \equiv \lnot \nc \lnot \varphi$.

We let $\mathcal{L}_{\nc, \gamma}$ denote the spatial fragment omitting the temporal $\FCircleOp, \BCircleOp$ operators.

The semantics is defined over dynamic topological systems, where the interpretations of variables range over some `admissible' subsets of our structures, yielding the corresponding models.
Regarding expressive power, the novelty of our proposal is to enrich dynamic topological logic with until-like reachability operators.
The intended semantics is that a point $x$ in a topological space $X$ satsifies $\gamma(\varphi,\psi)$ if there is a (continuous) path $\pi\colon [0,1] \to X$ such that $\pi(0) = x$, $\pi(t)$ satisfies $\varphi$ for $t\in (0,1)$, and $\pi(1)$ satisfies $\psi$.
We can think of this operation in algebraic terms by defining $\gamma\colon 2^  X\times 2 ^ X\to 2^X$, where $\gamma(A,B)$ is the set of all $ x\in X $ such that there exists a path $\pi\colon [0,1]$ with $\pi[(0,1)]\subseteq A$ and $\pi(1)  \in B$.
In order to allow for general semantics for languages with $\gamma$, the admissible valuations must be closed under this operation, leading to the following definition.

\begin{defi}
 Let $X$ be a topological space.
 A family $\mathcal A\subseteq 2^X$ is a {\em reachability algebra} if
 \begin{enumerate}

     \item $\mathcal A$ contains $X$ and is closed under (finite) Boolean operations

     \item If $A\in \mathcal A$ then $ \mathcal \IntOp(A) \in\mathcal A$

          \item If $A,B\in \mathcal A$ then $ \gamma (A,B) \in\mathcal A$.

 \end{enumerate}

Here $\IntOp(A)$ denotes topological interior of $A$.

An {\em h-dynamic reachability space} is a tuple $(X,\mathcal A,f)$, where $X$ is a topological space, $f\colon X\to X$ is a homeomorphism, and moreover, for every $A \subseteq X$, $A\in\mathcal A$ if and only if $f[A] \in \mathcal A$.\end{defi}

The prefix ``h'' in ``h-dynamic'' stands for ``homeomorphism''. In this paper we only deal with dynamical systems where the dynamics is invertible, i.e. the map $f\colon X\to X$ is a homeomorphism.

Note that $2^X$ is an example of a reachability algebra, but as we will see, so are polyhedral sets.

\begin{defi}[Models]
We define the following model structures over our languages:
\begin{itemize}
    \item A \emph{reachability model} is a triple $\mathcal{X} = (X, \mathcal A, V)$, where $\mathcal A$ is a reachability algebra and $V\colon \mathbb P\to \mathcal A$.

    \item An \emph{h-dynamic reachability model} is a tuple $\mathcal{X} = (X, \mathcal A, f, V)$, where $(X,\mathcal A,V)$ is a reachability model and $(X,\mathcal A,f)$ is an h-dynamic reachability space.
\end{itemize}
\end{defi}

We omit the parameter $\mathcal A$ in the above notation whenever $\mathcal A=2^X$.
Given a model $\mathcal{X}$ and an evaluation point $x$, the satisfaction relation $\models$ and the truth sets $\val{\varphi} = \{x \mid \mathcal{X}, x \models \varphi\}$ are defined inductively. For the boolean connectives, the clauses are standard. For the modal operators, we define:

\begin{itemize}
    \item $\mathcal{X}, x \models \nc \varphi \iff x \in \IntOp\val{\varphi}$.
    \item $\mathcal{X}, x \models \FCircleOp \varphi \iff f(x) \in \val{\varphi}$.
    \item $\mathcal{X}, x \models \BCircleOp \varphi \iff f^{-1}(x) \in \val{\varphi}$.
    \item $\mathcal{X}, x \models \gamma(\varphi, \psi) \iff$ there exists a continuous topological path $\pi: [0,1] \to X$ such that $\pi(0) = x$, $\pi(1) \in \val{\psi}$, and $\pi((0,1)) \subseteq \val{\varphi}$.
\end{itemize}

We are especially interested in reachability algebras where $\mathcal A$ is composed of geometrically nice, piecewise-linear sets. In the next section we recall the relevant definitions.

\section{The Algebras of Polyhedral Subsets}

The focus of this study, \emph{polyhedra}, can be seen as finite unions of \emph{polytopes}, i.e. convex hulls of finite subsets of an Euclidean space. Simplest such structures are points, line segments, triangles, tetrahedra, etc.

\begin{defi}[Simplex]
A \emph{$d$-simplex} $\sigma$ is the convex hull of a finite set $V = \{v_0,\ldots,v_d\} \subseteq \mathbb{R}^m$ of $d + 1$ affinely independent points. The value $d$ defines the \emph{dimension} of $\sigma$, and the vectors $v_0,\ldots,v_d$ are its \emph{vertices}.
\end{defi}

Simplices are intrinsically bounded, convex, and compact subspaces of $\mathbb{R}^m$. A \emph{face} of $\sigma$ is the convex hull $\tau$ of a non-empty subset $T \subseteq V$ of its vertex set. If $T$ is a proper subset of $V$, then $\tau$ is a \emph{proper face} of $\sigma$. We denote the set-theoretic inclusion partial order relation on simplices by $\tau \peq \sigma$.

\begin{defi}[Relative Interior / Open Cell]
The \emph{relative interior} of a simplex $\sigma$ with vertex set $\{v_0,\ldots,v_d\}$ is defined as:
\[
\tilde{\sigma} = \left\{ \sum_{i=0}^d \lambda_i v_i \ \mid \ \forall i, \lambda_i \in (0,1] \text{ and } \sum_{i=0}^d \lambda_i = 1 \right\}.
\]
An \emph{open cell} is the relative interior of some simplex.
\end{defi}

Every simplex $\sigma$ can be partitioned into the relative interiors of its faces.

More complex structures are generated by joining simplices along shared boundaries.

\begin{defi}[Simplicial Complex]
A \emph{simplicial complex} $K$ is a finite set of simplices in $\mathbb{R}^m$ satisfying the following conditions:
\begin{itemize}
    \item[\normal{1.}] If $\sigma \in K$ and $\tau \peq \sigma$, then $\tau \in K$.
    \item[\normal{2.}] If $\sigma, \tau \in K$, then either $\sigma \cap \tau = \emptyset$ or $\sigma \cap \tau \peq \sigma$.
\end{itemize}
\end{defi}

The underlying geometric space of $K$, called the \emph{polyhedron} of $K$ and denoted by $|K|$, is the set-theoretic union of its component simplices:
\[
|K| = \bigcup_{\sigma \in K} \sigma.
\]
Because $K$ is finite and each simplex is a compact subset of $\mathbb{R}^m$, every polyhedron $|K|$ investigated in this setting is a compact topological space under the subspace topology inherited from $\mathbb{R}^m$.

Distinct simplicial complexes can have the same underlying polyhedron. Conversely, a given  non-finite polyhedron can be subdivided into simplices in many ways. Whenever $P$ is a polyhedron and $K$ is a simplicial complex with $|K|=P$ we will say that $K$ triangulates $P$. Each triangulation induces a partition of $P$ into the open cells.

\begin{lemm}
Each point of a polyhedron $|K|$ belongs to the relative interior of exactly one simplex in $K$. That is, the collection of open cells $\tilde{K} = \{ \tilde{\sigma} \mid \sigma \in K\}$ forms a partition of $|K|$.
\end{lemm}

Rather than coupling semantics to one specific simplicial complex, we consider subsets that can be resolved by \emph{some} appropriate triangulation of a given space.

\begin{defi}[Polyhedral Set]
Let $P$ be a polyhedron. A subset $A \subseteq P$ is a \emph{polyhedral set} if there exists a simplicial complex $L$ such that $P = |L|$ and $A$ is a finite union of open cells from $\tilde{L}$. We denote the collection of all polyhedral subsets of $P$ by $\Poly{P}$.
\end{defi}

\begin{prop}\label{prop:closurealgebra}
Let $P$ be a polyhedron. The collection $\Poly{P}$ satisfies the following structural properties:
\begin{itemize}
    \item[\normal{1.}] $\Poly{P}$ forms a Boolean algebra under finite set-theoretic unions, intersections, and complements relative to $P$.
    \item[\normal{2.}] $\Poly{P}$ is closed under the topological interior ($\IntOp$) and closure ($\ClOp$) operators of $P$.
\end{itemize}
Consequently, $(\Poly{P}, \cup, \cap, \neg, \ClOp, \emptyset, P)$ is a closure algebra.
\end{prop}

Dynamic transitions over these spaces are modelled using functions that respect the underlying polyhedral structure.

\begin{defi}[Piecewise-Linear Map]
Let $P \subseteq \mathbb{R}^m$ and $Q \subseteq \mathbb{R}^n$ be polyhedra. A continuous function $f: P \to Q$ is a \emph{PL map} if its graph
\[
\Gamma(f) = \{(x, f(x)) \in P \times Q \mid x \in P\}
\]
is a subpolyhedron of the product space $\mathbb{R}^m \times \mathbb{R}^n$.
\end{defi}

By the standard theory of polyhedra and PL maps, a continuous map between polyhedra is PL exactly when, after suitable subdivisions, the map is affine-linear on each simplex. See, e.g. \cite{RourkeSanderson1972} for more details about polyhedra and PL-maps.

A fundamental property of PL maps is that they  preserve and reflect polyhedral sets.

\begin{lemm}\label{lem:plpreimage}
If $f: P \to Q$ is a PL map and $A \in \Poly{Q}$, then $f^{-1}(A) \in \Poly{P}$.
\end{lemm}

Although evaluating $\gamma(\varphi,\psi)$ in a polyhedron $P$ depends on arbitrary continuous images of the unit interval $[0,1]$, seeing that $[0,1]$ is itself a 1-simplex, it is natural to ask whether PL maps $\pi:[0,1]\to P$ suffice. Indeed this turns out the case, as the following lemmas confirm.

\begin{lemm}{\cite[Lemma 3.10]{BCGGLM22}}\label{lem:PL paths}
Let $P$ be a polyhedron and let $x,y\in P$. Then there exists a path from $x$ to $y$ in $P$ iff there is a PL path from $x$ to $y$ in $P$.
\end{lemm}

\begin{lemm}{\cite[Lemma 3.11]{BCGGLM22}}\label{lem:PL gamma}
Let $\mathcal{X}$ be a polyhedral model. Then $\mathcal{X},x\models\gamma(\varphi,\psi)\iff$ there exists a PL path $\pi: [0,1] \to P$ with $\pi(0) = x$, $\pi(1) \in \val{\psi}$, and $\pi((0,1)) \subseteq \val{\varphi}$.
\end{lemm}

It follows that if $A, B \in \Poly{P}$, then $\val{\gamma(A, B)} \in \Poly{P}$.

By combining Lemmas~\ref{lem:plpreimage},~\ref{lem:PL paths},~\ref{lem:PL gamma} and Proposition~\ref{prop:closurealgebra} we obtain that the structure $(P,\Poly{P}, f)$ with $P$ a polyhedron and $f$ a PL-homeomorphism is actually a special kind of an h-dynamic reachability space.

\begin{defi}[H-dynamic Polyhedral Space]
An \emph{h-dynamic polyhedral space} is a tuple $(P,\Poly{P}, f)$ where $P$ is a polyhedron and $f: P \to P$ is a PL homeomorphism.
\end{defi}

\section{The Axiomatic Systems}
The $\mathcal{L}_{\nc, \gamma}$-logics of topological spaces, of finite topological spaces and of polyhedral reachability algebras have been recently axiomatized. We recall the corresponding axiomatic systems and then extend them with axioms dealing with the dynamics, in the full language $\mathcal L$.

The reachability logic of finite topological spaces is axiomatized in \cite{PolyCompleteness}. The corresponding logic $\lgc{ALR}$ extends the modal system $\lgc{S4}$ with specific axioms and rules governing the behaviour of the reachability operator.

The system $\lgc{ALR}$ contains all propositional tautologies and the following axioms:
\begin{itemize}
    \item[\normal{K.}] $\nc(\varphi \to \psi) \to (\nc \varphi \to \nc \psi)$
    \item[\normal{T.}] $\nc \varphi \to \varphi$
    \item[\normal{4.}] $\nc \varphi \to \nc \nc \varphi$
    \item[\normal{$\gamma$1.}] $\ \ \psi \lor (\varphi \land \gamma(\varphi, \psi)) \to \nc(\varphi \to \gamma(\varphi, \psi))$
    \item[\normal{$\gamma$2.}] $\ \Diamond(\varphi \land \gamma(\varphi, \psi)) \to \gamma(\varphi, \psi)$
\end{itemize}
The rules of inference for $\lgc{ALR}$ are Modus Ponens, Necessitation $\left(\frac{\varphi}{\nc \varphi}\right)$, and two rules for reachability:
\[
\text{(R1) } \frac{\varphi \to \varphi' \quad \psi \to \psi'}{\gamma(\varphi, \psi) \to \gamma(\varphi', \psi')} \qquad\qquad \text{(R2) } \frac{\psi \to \nc(\varphi \to \psi) \quad \varphi \land \Diamond(\varphi \land \psi) \to \psi}{\gamma(\varphi, \psi) \to \nc(\varphi \land \psi)}
\]

The Polyhedral Logic of Reachability $\lgc{PLR}$ extends $\lgc{ALR}$ by the following axiom of Grzegorczyk:
\[\mathrm{Grz.   }\ \ \ \ \nc(\nc(\varphi \to \nc \varphi) \to \varphi) \to \nc \varphi\]

Completeness of $\lgc{PLR}$ and $\lgc{ALR}$ has been established in~\cite{PolyCompleteness}.

The reachability logic of all topological spaces has recently been axiomatized in~\cite{gagarin}. The corresponding system $\lgc{TLR}$ (Topological Logic of Reachability) is defined by all axioms and rules of $\lgc{S4}$ for $\nc$ together with the following:
\begin{itemize}
    \item[\normal{A0.}] $\ \ \gamma(\varphi,\psi_1 \lor \psi_2) \to \gamma(\varphi, \psi_1) \lor \gamma(\varphi \lor \psi_2)$
    \item[\normal{A1.}] $\ \ \gamma(\varphi, \varphi \land \gamma(\varphi, \psi) \to \gamma(\varphi,\psi))$
    \item[\normal{A2.}] $\ \ \varphi \to \gamma(\varphi,\varphi)$
    \item[\normal{A3.}] $\ \ \lnot \gamma(\bot, \top)$
    \item[\normal{A4.}] $\ \ \gamma(\varphi,\lnot \gamma(\varphi, \psi)) \to \lnot \psi$
    \item[\normal{A5.}] $\ \ \gamma(\varphi,\psi) \to \gamma(\varphi \land \gamma(\varphi, \psi), \psi)$
    \item[\normal{A6.}] $\ \ \nc \psi \land \gamma(\varphi, \chi) \to \gamma(\varphi \land \nc \psi, (\varphi \land \lnot \nc \psi) \lor \chi)$
    \item[\normal{A7.}] $\ \ \hat{\gamma}(\chi \land (\nc \varphi_1 \lor \varphi_2) \land (\hat{\gamma}(\chi \land \nc \varphi_1, \psi) \lor \hat{\gamma}(\chi \land \nc \varphi_2,\psi) \to \psi), \psi) \to \psi$
    \item[\normal{A8.}] $\ \ \gamma(\varphi \land \nc \lnot \varphi, \top) \to \varphi$
    \item[\normal{R1.}] $\ \ \ \ \ \ \frac{\varphi \to \varphi' \quad \psi \to \psi'}{\gamma(\varphi, \psi) \to \gamma(\varphi', \psi')}$
\end{itemize}

where $\hat{\gamma}(\varphi,\psi) := \varphi \land \gamma(\varphi, \varphi \land \psi)$.

\begin{defi}

    Let $\mathfrak T$ denote the class of all  reachability spaces $(X,2^X)$;

Let $\mathfrak F$ denote the class of all \emph{finite} reachability spaces;

Let $\mathfrak A$ denote the class of all \emph{Alexandroff} reachability spaces;

Let $\mathfrak P$ denote the class of all \emph{polyhedral} reachability spaces $(P,\Poly{P})$.
\end{defi}

The following holds:

\begin{thm}\label{t:comp_no-hD}
\quad
\begin{enumerate}
\item $\sf TLR$ is sound and complete for $\mathfrak T$ \cite{gagarin}.

\item $\sf ALR$ is sound and complete both for $\mathfrak F$ and for $\mathfrak A$ \cite{PolyCompleteness}.
\item $\sf PLR$ is sound and complete for $\mathfrak P$ \cite{PolyCompleteness}.

\end{enumerate}
\end{thm}

A variant of dynamic polyhedral logic without a reachability operator was introduced in~\cite{Kopnev},  where the temporal dynamics was modeled using relations. We now define the classes of h-dynamic structures corrisponding to the considered logics.

\begin{defi}
    Let $\mathfrak{hDT}$ denote the class of all  h-dynamic reachability spaces $(X,2^X,f)$, where $f:X\to X$ is a homeomorphism;

Let $\mathfrak {hDF}$ denote the class of all \emph{finite} h-dynamic reachability spaces;

Let $\mathfrak{hDA}$ denote the class of all \emph{Alexandroff} h-dynamic reachability spaces;

Let $\mathfrak{hDP}$ denote the class of all  \emph{h-dynamic polyhedral} spaces $(P,\Poly{P},f)$ where $f:P\to P$ is a PL-homeomorphism.
\end{defi}
To axiomatize the logics of these classes in the language $\mathcal L$ we add the temporal and spatio-temporal axioms to the corresponding logics.

\begin{defi}\label{d:hD axioms}
Given a logic $\Lambda \in \{ \lgc{TLR, ALR, PLR} \}$ in the language $\mathcal{L}_{\nc, \gamma}$, we define by $\dyn{\Lambda}$ the h-dynamic extension of $\Lambda$ in the language $\mathcal{L}$, obtained by adding the following  set of axioms and rules:
\begin{itemize}
    \item[\normal{$\FCircleOp$1.}] $\ \ \ \lnot \FCircleOp \varphi \leftrightarrow \FCircleOp \lnot \varphi$
    \item[\normal{$\FCircleOp$2.}] $\ \ \ \FCircleOp(\varphi \land \psi) \leftrightarrow \FCircleOp \varphi \land \FCircleOp \psi$
    \item[\normal{$\FCircleOp$3.}] $\ \ \ \gamma(\FCircleOp \varphi, \FCircleOp \psi) \to \FCircleOp \reach{\varphi}{\psi}$
    \item[\normal{$\FCircleOp$4.}] $\ \ \ \nc\FCircleOp\varphi\to \FCircleOp\nc\varphi$

    \item[\normal{$\BCircleOp$1.}] $\ \ \ \lnot \BCircleOp \varphi \leftrightarrow \BCircleOp \lnot \varphi$
    \item[\normal{$\BCircleOp$2.}] $\ \ \ \BCircleOp(\varphi \land \psi) \leftrightarrow \BCircleOp \varphi \land \BCircleOp \psi$
    \item[\normal{$\BCircleOp$3.}] $\ \ \ \gamma(\BCircleOp \varphi, \BCircleOp \psi) \to \BCircleOp \reach{\varphi}{\psi}$
    \item[\normal{$\BCircleOp$4.}]$\ \ \ \nc\BCircleOp\varphi\to \BCircleOp\nc\varphi$

    \item[\normal{$\FCircleOp\BCircleOp$1.}]\ \ \ \ \ 
    $\FCircleOp\BCircleOp\varphi\leftrightarrow\varphi$

    \item[\normal{$\BCircleOp\FCircleOp$1.}]\ \ \ \ \ 
    $\BCircleOp\FCircleOp\varphi\leftrightarrow\varphi$

    \item[\normal{$\FCircleOp$N}] $\ \ \ \ \ \frac{\varphi}{\FCircleOp\varphi}$

    \item[\normal{$\BCircleOp$N}] $\ \ \ \ \ \frac{\varphi}{\BCircleOp\varphi}$

\end{itemize}
\end{defi}

To show soundness of these systems, we need a couple of lemmas.

\begin{lemm}\label{lem:sound_homeo}
Let $\mathcal X=(X,\mathcal A,f)$ be an h-dynamic reachability space. Then:
\begin{itemize}
    \item $\mathcal X \models \gamma(\FCircleOp \varphi, \FCircleOp \psi) \to \FCircleOp \reach{\varphi}{\psi}$,
    
    \item $ \mathcal X\models \gamma(\BCircleOp \varphi, \BCircleOp \psi) \to \BCircleOp \reach{\varphi}{\psi}$,

    \item $ \mathcal X\models \nc\FCircleOp\varphi\to \FCircleOp\nc\varphi$,

    \item $ \mathcal X\models \nc\BCircleOp\varphi\to \BCircleOp\nc\varphi$.
\end{itemize}
\end{lemm}

\begin{proof}
We only consider the first claim, since the second one is symmetric and the last two are easily seen to follow from $f$ and $f^{-1}$ being continuous maps.

Suppose $\mathcal X, x \models \gamma(\FCircleOp\varphi, \FCircleOp\psi)$ for a valuation $\val{\cdot}$. Then there exists a topological path $\pi: [0,1] \to X$ such that $\pi(0) = x$, $\pi(1) \in \val{\FCircleOp\psi}$, and $\pi((0,1)) \subseteq \val{\FCircleOp\varphi}$. By definition, $f(\pi(1)) \in \val{\psi}$ and $f(\pi((0,1))) \subseteq \val{\varphi}$. Consider the path $\pi' = f \circ \pi$. Since $f$ is continuous, $\pi'$ is a continuous topological path. We have $\pi'(0) = f(x)$, $\pi'(1) \in \val{\psi}$, and $\pi'((0,1)) \subseteq \val{\varphi}$, meaning $f(x) \in \val{\reach{\varphi}{\psi}}$. Thus $x \in \val{\FCircleOp\reach{\varphi}{\psi}}$.
\end{proof}

\begin{lemm}\label{lemHomAx}
Let $\Lambda$ be an h-dynamic logic containing the axioms and rules of Definition~\ref{d:hD axioms}. Then:
\begin{itemize}
\item $\Lambda\vdash \gamma(\FCircleOp \varphi, \FCircleOp \psi) \leftrightarrow \FCircleOp \reach{\varphi}{\psi}$,

\item $\Lambda\vdash \gamma(\BCircleOp \varphi, \BCircleOp \psi) \leftrightarrow \BCircleOp \reach{\varphi}{\psi}$

\item $\Lambda\vdash\nc\FCircleOp\varphi\leftrightarrow \FCircleOp\nc\varphi$,

\item $\Lambda\vdash\nc\BCircleOp\varphi\leftrightarrow \BCircleOp\nc\varphi$.
\end{itemize}
\end{lemm}

\begin{proof}
Let us first focus on $\gamma(\FCircleOp \varphi, \FCircleOp \psi) \leftrightarrow \FCircleOp \reach{\varphi}{\psi}$ and note that we only need to show $\FCircleOp \reach{\varphi}{\psi} \to \gamma(\FCircleOp \varphi, \FCircleOp \psi) $, as the other implication is an instance of \normal{$\FCircleOp$3}.
The derivation is as follows:
\begin{enumerate}
    \item Assume $\FCircleOp \reach{\varphi}{\psi}$.
    \item By \normal{$\BCircleOp\FCircleOp$1}, this is equivalent to $\FCircleOp \reach{\BCircleOp\FCircleOp\varphi}{\BCircleOp\FCircleOp\psi}$.
    \item By \normal{$\BCircleOp$3}, this implies $\FCircleOp \BCircleOp\reach{ \FCircleOp\varphi}{ \FCircleOp\psi}$.
    \item By \normal{$\FCircleOp\BCircleOp$1}, this is equivalent to $  \reach{ \FCircleOp\varphi}{ \FCircleOp\psi}$.
\end{enumerate}

Similarly, assuming $\FCircleOp\nc\varphi$, we use \normal{$\BCircleOp\FCircleOp$1} to get $\FCircleOp\nc\BCircleOp\FCircleOp\varphi$, then derive $\FCircleOp\BCircleOp\nc\FCircleOp\varphi$ by \normal{$\BCircleOp$4} and finally, use \normal{$\FCircleOp\BCircleOp$1} to arrive at $\nc\FCircleOp\varphi$. Hence $\Lambda\vdash\FCircleOp\nc\varphi\to\nc\FCircleOp\varphi$. The converse implication is simply an instance of \normal{$\FCircleOp$4}.

The other two cases are treated in a symmetric fashion.
\end{proof}

\begin{thm}\label{thm:soundness}
The axiomatic systems $\dyn{\lgc{TLR}}$ $\dyn{\lgc{ALR}}$ are sound with respect to the classes $\mathfrak {hDT}$, $\mathfrak {hDF}$ and $\mathfrak {hDA}$, while
the axiomatic system $\dyn{\lgc{PLR}}$ is sound with respect to the class $\mathfrak {hDP}$.
\end{thm}

\begin{proof}
Soundness of the pure spatial axioms and rules follows from Theorem~\ref{t:comp_no-hD}. The functionality of $f$ ensures that $\FCircleOp$ commutes with boolean connectives. Soundness of the interaction axioms follows from Lemma~\ref{lem:sound_homeo}.
\end{proof}

In the next section we establish our main findings regarding the completeness of these axiomatic systems.

\section{Completeness of h-Dynamic Reachability Logics}

We show that our h-dynamic reachability logics are complete for their
respective classes of models by reducing formulas to equivalent
reachability configurations through a translation function that pushes
temporal steps to the atomic layer.

By the axioms $\FCircleOp\BCircleOp 1$ and
$\BCircleOp\FCircleOp 1$ of Definition~\ref{d:hD axioms}, every finite
sequence of forward and backward temporal operators is provably
equivalent to a sequence consisting entirely of forward operators or
entirely of backward operators. Its length and direction are determined
by the net temporal displacement. As usual, for $n\geq 0$, we let
$\FCircleOp^n\varphi$ denote $\varphi$ preceded by $n$ instances of
$\FCircleOp$. We however also allow for $n<0$, in which case we let $\FCircleOp^n\varphi$ denote $\varphi$
preceded by $-n$ instances of $\BCircleOp$. With this convention, the
following transformation sends every formula of $\mathcal L$ to a
provably equivalent formula in which temporal operators occur only in
front of propositional variables.

\begin{defi}
Let $\varphi\in\mathcal L$. The structural transformation
$g(\varphi)$ is defined inductively as follows:
\begin{itemize}
    \item
    $g(\FCircleOp^n p)=\FCircleOp^n p$
    for $p\in\Prop$;

    \item
    $g(\FCircleOp^n\lnot\varphi)
       =\lnot g(\FCircleOp^n\varphi)$;

    \item
    $g\bigl(\FCircleOp^n(\varphi\land\psi)\bigr)
       =g(\FCircleOp^n\varphi)
        \land g(\FCircleOp^n\psi)$;

    \item
    $g(\FCircleOp^n\nc\varphi)
       =\nc g(\FCircleOp^n\varphi)$;

    \item
    $g\bigl(\FCircleOp^n\gamma(\varphi,\psi)\bigr)
       =\gamma\bigl(
          g(\FCircleOp^n\varphi),
          g(\FCircleOp^n\psi)
        \bigr)$;

    \item
    $g(\FCircleOp^n\BCircleOp\varphi)
       =g(\FCircleOp^{n-1}\varphi)$
    for $n>0$;

    \item
    $g(\FCircleOp^n\FCircleOp\varphi)
       =g(\FCircleOp^{n+1}\varphi)$
    for $n<0$.
\end{itemize}
Except in the last two clauses, $n$ may be any integer, including $0$.
In particular, taking $n=0$ shows that $g$ commutes with negation and
$\nc$, and distributes over conjunction and $\gamma$ (and hence over
the derived Boolean connectives).

For example,
\begin{align*}
g\bigl(\BCircleOp(\neg p\land\nc q)\bigr)
   &=\neg\BCircleOp p\land\nc\BCircleOp q,\\
g\bigl(\FCircleOp(p\to\BCircleOp q)\bigr)
   &=\FCircleOp p\to q,\\
g\bigl(\BCircleOp\gamma(p,\BCircleOp \Box q)\bigr)
   &=\gamma(\BCircleOp p,\Box\BCircleOp\BCircleOp q).
\end{align*}
The first example illustrates how temporal operators are pushed through
the spatial and Boolean structure. The second illustrates cancellation
of opposite temporal shifts, while the third illustrates accumulation
of shifts in the same direction. Thus, $g$ preserves the purely spatial
structure of a formula while combining the temporal shifts along each
branch into a single shift attached to the corresponding atom.

Formulas of the form $\FCircleOp^n p$ are called \emph{near-atoms}.
A formula is \emph{simple} if it belongs to the range of $g$.
Equivalently, a formula is simple iff it is built from near-atoms using
the Boolean connectives, $\nc$, and $\gamma$.
\end{defi}

Simple formulas can be parsed as standard formulas of the spatial language $\mathcal{L}_{\nc, \gamma}$ by treating near-atoms $\FCircleOp^n p$ as distinct propositional variables.

\begin{lemm}\label{lem:g_equivalence}
Let $\varphi\in\mathcal L$ and
$\Lambda\in\{\lgc{TLR},\lgc{ALR},\lgc{PLR}\}$. Then
\[
\dyn{\Lambda}\vdash\varphi
\quad\Longleftrightarrow\quad
\Lambda\vdash g(\varphi),
\]
where, on the right-hand side, the near-atoms occurring in
$g(\varphi)$ are regarded as distinct propositional variables.
\end{lemm}

\begin{proof}
For the left-to-right implication, suppose that
$\dyn{\Lambda}\vdash\varphi$. We prove
$\Lambda\vdash g(\varphi)$ by induction on the length of a derivation
of $\varphi$ in $\dyn{\Lambda}$. Since $g$ preserves the Boolean and
spatial structure, it sends every axiom of $\Lambda$ to an instance of
the same axiom and preserves every inference rule of $\Lambda$.

The temporal axioms are sent to propositional tautologies. For example,
\[
g(\lnot\FCircleOp\psi\leftrightarrow
  \FCircleOp\lnot\psi)
=
\lnot g(\FCircleOp\psi)\leftrightarrow
\lnot g(\FCircleOp\psi),
\]
and
\[
\begin{split}
g\bigl(
  \gamma(\FCircleOp\psi,\FCircleOp\chi)
  \to\FCircleOp\gamma(\psi,\chi)
\bigr)
={}&
\gamma\bigl(g(\FCircleOp\psi),g(\FCircleOp\chi)\bigr)
\to
\gamma\bigl(g(\FCircleOp\psi),g(\FCircleOp\chi)\bigr).
\end{split}
\]
Likewise,
\[
g(\nc\FCircleOp\psi\to\FCircleOp\nc\psi)
=
\nc g(\FCircleOp\psi)\to\nc g(\FCircleOp\psi).
\]
The backward temporal axioms are treated symmetrically, and both
inverse axioms are sent to formulas of the form
$g(\psi)\leftrightarrow g(\psi)$.

Modus ponens is preserved because
\[
g(\psi\to\chi)=g(\psi)\to g(\chi),
\]
and $\nc$-necessitation is preserved because $g$ commutes with $\nc$.

It remains to consider temporal necessitation. Suppose $\FCircleOp\psi$ is inferred from $\psi$. By the induction hypothesis,
$\Lambda\vdash g(\psi)$. The formula $g(\FCircleOp\psi)$ is obtained
from $g(\psi)$ by uniformly replacing every near-atom
$\FCircleOp^k p$ by $\FCircleOp^{k+1}p$. Regarding near-atoms as
propositional variables and using closure under uniform substitution,
we obtain
\[
\Lambda\vdash g(\FCircleOp\psi).
\]
The case of $\BCircleOp$-necessitation is analogous: every near-atom
$\FCircleOp^k p$ is replaced by $\FCircleOp^{k-1}p$.

To prove the right-to-left implication, we first observe that
\[
\dyn{\Lambda}\vdash\varphi\leftrightarrow g(\varphi).
\tag{$*$}
\]
This is proved by induction on the construction of $\varphi$. The Boolean cases follow from the temporal Boolean axioms, while the $\nc$- and $\gamma$-cases follow from Lemma~\ref{lemHomAx}. Mixed temporal prefixes
are reduced using the inverse axioms
\normal{$\FCircleOp\BCircleOp 1$} and
\normal{$\BCircleOp\FCircleOp 1$}. Iterating these equivalences gives the result for arbitrary finite temporal prefixes.

Suppose now that $\Lambda\vdash g(\varphi)$, with its
near-atoms regarded as distinct propositional variables. Since
$\dyn{\Lambda}$ extends $\Lambda$ and is closed under uniform
substitution, it follows that
\[
\dyn{\Lambda}\vdash g(\varphi).
\]
Combining this with $(*)$, we conclude that
$\dyn{\Lambda}\vdash\varphi$.
\end{proof}

To establish completeness, we show how to construct an h-dynamic model
of $\varphi$ from a model of $g(\varphi)$. The number of copies required
in this construction is determined by the maximal temporal depth of
$\varphi$.

\begin{defi}
The \emph{duration} $d(\varphi)$ of a formula
$\varphi\in\mathcal L$ is defined inductively by
\begin{align*}
d(p) &= 0,\\
d(\lnot\varphi) &= d(\varphi),\\
d(\varphi\land\psi)
  =d(\gamma(\varphi,\psi))
  &=\max\{d(\varphi),d(\psi)\},\\
d(\nc\varphi) &= d(\varphi),\\
d(\FCircleOp\varphi)
  =d(\BCircleOp\varphi)
  &=d(\varphi)+1.
\end{align*}
\end{defi}

\begin{lemm}\label{lemmGtoMod}
If $g(\varphi)$ is satisfiable in a topological model $\mathcal X$,
then $\varphi$ is satisfiable in an h-dynamic topological model
$\mathcal X'$. Moreover, if $\mathcal X$ is finite or Alexandroff, then $\mathcal X'$ also has the corresponding property.
\end{lemm}

\begin{proof}
Let
\[
n=\max\{1,d(\varphi)\}
\qquad\text{and}\qquad
I_n=\{-n,\ldots,-1,0,1,\ldots,n\},
\]
where $I_n$ is equipped with the discrete topology. Suppose that
$\mathcal X=(X,V)$ and
\[
\mathcal X,a\models g(\varphi)
\]
for some $a\in X$, with the near-atoms occurring in $g(\varphi)$
treated as propositional variables.

Let
\[
X'=X\times I_n
\]
with the product topology, and define $f':X'\to X'$ by
\[
f'(x,k)=
\begin{cases}
(x,k+1), & k<n,\\
(x,-n),  & k=n.
\end{cases}
\]
Thus, $f'$ cyclically permutes the copies of $X$ and is therefore a
homeomorphism. Define a valuation $V'$ on $X'$ by
\[
(x,k)\in V'(p)
\quad\Longleftrightarrow\quad
x\in V(\FCircleOp^k p).
\]

We claim that, for every subformula $\psi$ of $\varphi$ and every
$k\in I_n$ such that
\[
|k|+d(\psi)\leq n,
\]
we have
\[
\mathcal X',(x,k)\models\psi
\quad\Longleftrightarrow\quad
\mathcal X,x\models g(\FCircleOp^k\psi).
\tag{$\dagger$}
\]
The claim is proved by induction on the construction of $\psi$. The
atomic case follows directly from the definition of $V'$, and the
Boolean cases are immediate.

For the spatial cases, observe that every subspace $X\times\{k\}$ is a clopen copy of $X$. Consequently, interiors are computed
within individual copies. Moreover, every path in $X'$ remains in a single copy, since its projection onto the discrete space $I_n$ has to be constant. The induction hypothesis therefore yields $(\dagger)$ for both $\nc$ and $\gamma$.

For the forward temporal case, the induction hypothesis gives
\[
\begin{aligned}
\mathcal X',(x,k)\models\FCircleOp\psi
&\Longleftrightarrow
\mathcal X',(x,k+1)\models\psi\\
&\Longleftrightarrow
\mathcal X,x\models g(\FCircleOp^{k+1}\psi)\\
&\Longleftrightarrow
\mathcal X,x\models g(\FCircleOp^k\FCircleOp\psi).
\end{aligned}
\]
The bound in $(\dagger)$ ensures that no cyclic wraparound occurs.
The backward temporal case is analogous, with $k-1$ in place of
$k+1$. This proves the claim.

Taking $k=0$ and $\psi=\varphi$ in $(\dagger)$ yields
\[
\mathcal X',(a,0)\models\varphi.
\]
Finally, $X'=X\times I_n$ is finite whenever $X$ is finite and is
Alexandroff whenever $X$ is Alexandroff, since $I_n$ is finite and
discrete.
\end{proof}

The next lemma is essentially analogous, but intended for polyhedral models, where we can make the map be geometrically `nicer'.

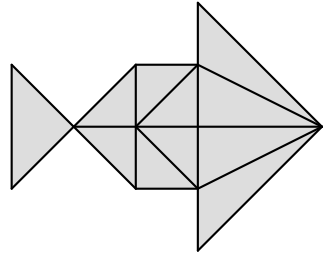
\begin{figure}[ht]
\centering
\begin{tikzpicture}[
    x=0.82cm,
    y=0.82cm,
    line cap=round,
    line join=round
]

\coordinate (Lt) at (0,3);
\coordinate (Lb) at (0,1);

\coordinate (A)  at (1,2);

\coordinate (B)  at (2,3);
\coordinate (M)  at (2,2);
\coordinate (F)  at (2,1);

\coordinate (T)  at (3,4);
\coordinate (Cc) at (3,3);
\coordinate (D)  at (3,2);
\coordinate (E)  at (3,1);
\coordinate (U)  at (3,0);

\coordinate (N)  at (5,2);

\fill[gray!70,opacity=0.38]
    (Lt) -- (Lb) -- (A) -- (F) -- (E) -- (U) --
    (N) -- (T) -- (Cc) -- (B) -- (A) -- cycle;

\draw[thick] (Lt) -- (Lb);
\draw[thick] (Lt) -- (A);
\draw[thick] (Lb) -- (A);

\draw[thick] (A) -- (B);
\draw[thick] (A) -- (F);
\draw[thick] (B) -- (Cc);
\draw[thick] (F) -- (E);

\draw[thick] (A) -- (N);
\draw[thick] (T) -- (U);

\draw[thick] (T) -- (N);
\draw[thick] (U) -- (N);
\draw[thick] (Cc) -- (N);
\draw[thick] (E) -- (N);

\draw[thick] (B) -- (M);
\draw[thick] (M) -- (E);
\draw[thick] (M) -- (F);
\draw[thick] (M) -- (Cc);

\end{tikzpicture}
\caption{The triangulated polyhedron $P$ used in the
two-dimensional illustration of Lemma~\ref{lem:poly_reconstruction}.}
\label{fig:polyhedron-P}
\end{figure}

Figures~\ref{fig:polyhedron-P} and~\ref{fig:polyhedron-P-prime} illustrate the
geometric idea in dimension two.  The first shows a triangulated polyhedron
$P$.  The second shows a five-fold instance of the construction, sufficient for formulas with duration $\leq 2$, in which a translated copy $P_0$ of $P$ is rotated around the
$x_2$-axis to obtain $P'=\bigcup_{k=-2}^{2}P_k$, where $P_k=f^k(P_0)$ and $f:P'\to P'$ is rotation through $72^\circ$.

\begin{figure}[hb]
\centering
\begin{tikzpicture}[
    scale=0.72,
    transform shape,
    line cap=round,
    line join=round,
    >=Latex,
    x={(1.2cm,0cm)},
    y={(0.06cm,0.60cm)},
    z={(0cm,1.2cm)}
]

\def\R{5}
\def\H{4}

\def\fx{2}
\def\fz{0.75}
\def\fs{0.55}

\coordinate (O) at (0,0,0);
\coordinate (Ctop) at (0,0,\H);

\coordinate (Xax) at (6,0,0);
\coordinate (Yax) at (0,6,0);
\coordinate (Zax) at (0,0,5.2);

\draw[blue,->,thick]
    (O) -- (Xax)
    node[below right] {$x_1$};

\draw[blue, ->,thick]
    (O) -- (Yax)
    node[above] {$x_0$};

\draw[blue, ->,thick]
    (O) -- (Zax)
    node[above] {$x_2$};

\node[below left] at (O) {$O$};

\foreach \k in {0,...,4}{
    \pgfmathsetmacro{\ang}{360/5*\k}

    \coordinate (QA\k) at
        ({\R*cos(\ang)},
         {\R*sin(\ang)},
         0);

    \coordinate (QB\k) at
        ({\R*cos(\ang)},
         {\R*sin(\ang)},
         \H);
}

\foreach \k in {1,...,4}{
    \fill[gray!50,opacity=0.12]
        (O) -- (QA\k) -- (QB\k) -- (Ctop) -- cycle;
}

\fill[gray!60,opacity=0.22]
    (O) -- (QA0) -- (QB0) -- (Ctop) -- cycle;

\foreach \k in {1,...,4}{
    \draw[dashed,gray!75,thin]
        (O) -- (QA\k) -- (QB\k) -- (Ctop);
}

\draw[gray!75,thin]
    (O) -- (QA0) -- (QB0) -- (Ctop);

\foreach \k in {0,...,4}{

    \pgfmathsetmacro{\ang}{360/5*\k}

    \foreach \P/\u/\v in {
        Lt/0/3,
        Lb/0/1,
        A/1/2,
        B/2/3,
        M/2/2,
        F/2/1,
        T/3/4,
        Cc/3/3,
        D/3/2,
        E/3/1,
        U/3/0,
        N/5/2
    }{
        \pgfmathsetmacro{\radial}{\fx+\fs*\u}
        \pgfmathsetmacro{\vertical}{\fz+\fs*\v}

        \coordinate (\P\k) at
            ({\radial*cos(\ang)},
             {\radial*sin(\ang)},
             \vertical);
    }

    \fill[gray!70,opacity=0.38]
        (Lt\k) -- (Lb\k) -- (A\k) -- (F\k) -- (E\k) --
        (U\k) -- (N\k) -- (T\k) -- (Cc\k) -- (B\k) --
        (A\k) -- cycle;

    \draw[thick] (Lt\k) -- (Lb\k);
    \draw[thick] (Lt\k) -- (A\k);
    \draw[thick] (Lb\k) -- (A\k);

    \draw[thick] (A\k) -- (B\k);
    \draw[thick] (A\k) -- (F\k);
    \draw[thick] (B\k) -- (Cc\k);
    \draw[thick] (F\k) -- (E\k);

    \draw[thick] (A\k) -- (N\k);
    \draw[thick] (T\k) -- (U\k);

    \draw[thick] (T\k) -- (N\k);
    \draw[thick] (U\k) -- (N\k);
    \draw[thick] (Cc\k) -- (N\k);
    \draw[thick] (E\k) -- (N\k);

    \draw[thick] (B\k) -- (M\k);
    \draw[thick] (M\k) -- (E\k);
    \draw[thick] (M\k) -- (F\k);
    \draw[thick] (M\k) -- (Cc\k);
}

\end{tikzpicture}
\caption{The rotational construction for $n=2$:
$P'=\bigcup_{k=-2}^{2}P_k$, where $P_{k+1}=f(P_k)$ and $f$ is a
$72^\circ$ rotation around the $x_2$-axis. The transparent
half-planes are included only to make the rotational geometry visible.}
\label{fig:polyhedron-P-prime}
\end{figure}
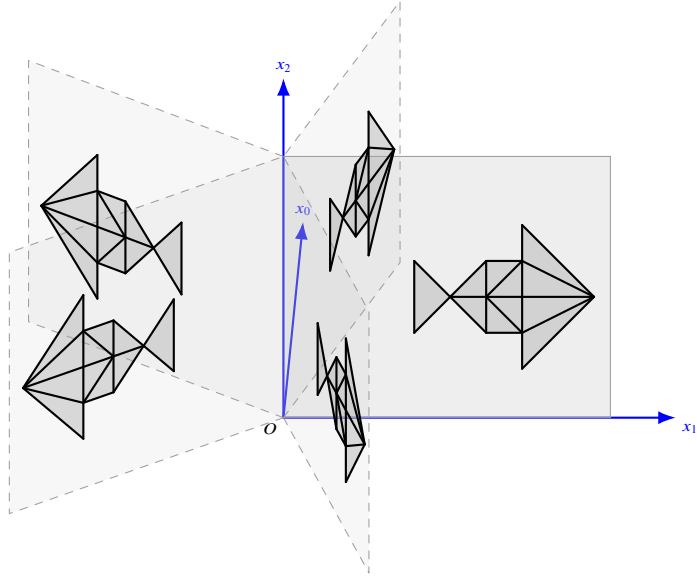

\begin{lemm}\label{lem:poly_reconstruction}
If $g(\varphi)$ is satisfiable in a polyhedral model, then $\varphi$
is satisfiable in an h-dynamic polyhedral model whose dynamics is a
rotation.
\end{lemm}

\begin{proof}
Suppose that $\mathcal M=(P,V)$ is a polyhedral model and
\[
\mathcal M,x_0\models g(\varphi),
\]
where the near-atoms $\FCircleOp^k p$ are treated as independent
propositional variables. Let $P\subseteq\mathbb R^m$ and put

\[
n=d(\varphi),
\qquad
\theta=\frac{2\pi}{2n+1}.
\]

Let $\mathbb R^{m+1}$ have coordinates
$(x_0,x_1,\ldots,x_m)$. Identify $P\subseteq\mathbb R^m$ with
$\{0\}\times P\subseteq\mathbb R^{m+1}$. 

Since $P$ is bounded, we may choose $a>0$ so that
\[
P_0:=a\bm e_1+\{0\}\times P
\]
is contained in the open half-hyperplane $H_0=\{(x_0,\ldots,x_m):x_0=0,\ x_1>0\}$. Write $\iota(\bm x)=a\bm e_1+(0,\bm x)$ for the resulting embedding of $P$ into $P_0$.

It will be convenient to represent vectors in $\mathbb R^{m+1}$ as
$(\bm u,\bm v)$, where $\bm u\in\mathbb R^2$ and
$\bm v\in\mathbb R^{m-1}$. 

Let $A$ be the $2\times2$ rotation matrix
through the angle $\theta$, and define
\[
R(\bm u,\bm v)=(A\bm u,\bm v).
\]
Thus, $R$ is a rotation around the axis $L=\{x_0=x_1=0\}$.

For $-n\leq k\leq n$, let
\[
P_k=R^k(P_0).
\]
The half-hyperplanes $R^k(H_0)$ are pairwise disjoint outside $L$.
Since $P_0\cap L=\varnothing$, and hence $P_k\cap L=\varnothing$ for
every $k$, the polyhedra $P_k$ are pairwise disjoint. Therefore,
\[
P'=\bigcup_{k=-n}^{n}P_k
\]
is a polyhedron.

Since $R^{2n+1}$ is the identity, $R$ cyclically permutes the sets
$P_k$. Consequently, its restriction
\[
f=R|_{P'}:P'\to P'
\]
is a PL homeomorphism.

Define a valuation $V'$ on $P'$ by
\[
R^k(\iota(\bm x))\in V'(p)
\quad\Longleftrightarrow\quad
\bm x\in V(\FCircleOp^k p),
\qquad -n\leq k\leq n.
\]
Each $V'(p)$ is a finite union of rotated copies of polyhedral subsets
of $P$, and hence belongs to $\Poly{P'}$.

We claim that, for every subformula $\psi$ of $\varphi$ and every
integer $k$ such that
\[
|k|+d(\psi)\leq n,
\]
we have
\[
(P',f,V'),R^k(\iota(\bm x))\models\psi
\quad\Longleftrightarrow\quad
(P,V),\bm x\models g(\FCircleOp^k\psi).
\tag{$\dagger$}
\]
The claim is proved by induction on the construction of $\psi$. The
atomic and Boolean cases follow directly from the definitions.

For the spatial cases, observe that the pairwise disjoint compact
polyhedra $P_k$ are open and closed in $P'$. Thus, interiors are
computed within individual copies. Moreover, every path in $P'$ lies
in a single $P_k$, since the continuous image of the connected space
$[0,1]$ cannot meet two distinct clopen components. The induction
hypothesis therefore yields $(\dagger)$ for both $\nc$ and $\gamma$.

For the forward temporal case, the induction hypothesis gives
\[
\begin{aligned}
(P',f,V'),R^k(\iota(\bm x))\models\FCircleOp\psi
&\Longleftrightarrow
(P',f,V'),R^{k+1}(\iota(\bm x))\models\psi\\
&\Longleftrightarrow
(P,V),\bm x\models g(\FCircleOp^{k+1}\psi)\\
&\Longleftrightarrow
(P,V),\bm x\models g(\FCircleOp^k\FCircleOp\psi).
\end{aligned}
\]
The bound in $(\dagger)$ ensures that no cyclic wraparound occurs.
The backward temporal case is analogous, with $k-1$ in place of
$k+1$. This proves the claim.

Taking $k=0$, $\psi=\varphi$, and $\bm x=x_0$ in $(\dagger)$ yields
\[
(P',f,V'),\iota(x_0)\models\varphi,
\]
as required.
\end{proof}

\begin{remark}
The passage from $\mathbb R^m$ to $\mathbb R^{m+1}$ in the proof is
used only to make the geometry of the construction particularly
transparent; it is not essential, provided that the ambient dimension
is at least two. Indeed, one may choose an axis of rotation in
$\mathbb R^m$ and place a suitably scaled or translated copy of $P$
sufficiently far from that axis. Since $P$ is bounded, it may be
enclosed in a ball whose angular diameter, as viewed from the rotation
axis, is strictly smaller than
\[
\frac{2\pi}{2n+1}.
\]
Equivalently, the copy of $P$ can be placed inside the interior of a
cone whose aperture is smaller than the angle of rotation. Rotating
this cone through successive multiples of $2\pi/(2n+1)$ produces
$2n+1$ pairwise disjoint cones, and hence pairwise disjoint rotated
copies of $P$. The remainder of the construction then proceeds exactly
as in the proof of Lemma~\ref{lem:poly_reconstruction}. The dimension-raising construction used in the proof has the advantage
of being uniform and of working independently of the dimension of the
original embedding. The same-dimensional construction, however, 
admits an even simpler geometric visualization, as illustrated in
Figure~\ref{fig:rotation-same-dimension}.
\end{remark}

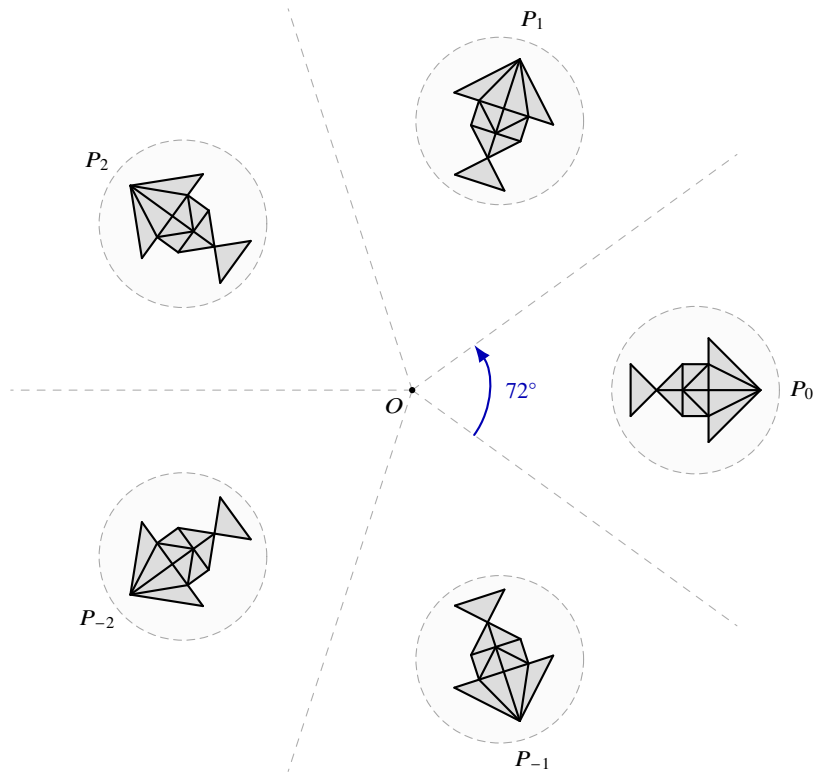
\begin{figure}[ht]
\centering
\begin{tikzpicture}[
    scale=0.78,
    line cap=round,
    line join=round,
    >=Latex
]

\def\OrbitR{4.8}   
\def\FishScale{0.44}
\def\BallR{1.42}   
\def\ConeR{6.8}    

\coordinate (O) at (0,0);


\foreach \j in {0,...,4}{
    \pgfmathsetmacro{\boundary}{36+72*\j}
    \draw[dashed,gray!65,thin]
        (O) -- ({\boundary}:\ConeR);
}

\fill (O) circle[radius=1.5pt];
\node[below left] at (O) {$O$};

\draw[->,blue!70!black,thick]
    (1.05,-0.76)
    arc[start angle=-36,end angle=36,radius=1.30];

\node[blue!70!black,right]
    at (1.45,0) {$72^\circ$};

%
\foreach \k/\idx in {
    0/0,
    1/1,
    2/2,
    3/{-2},
    4/{-1}
}{
    \pgfmathsetmacro{\ang}{72*\k}

    \coordinate (C\k) at
        ({\OrbitR*cos(\ang)},
         {\OrbitR*sin(\ang)});

    \fill[gray!30,opacity=0.10]
        (C\k) circle[radius=\BallR];

    \draw[densely dashed,gray!70,thin]
        (C\k) circle[radius=\BallR];

    \foreach \P/\u/\v in {
        Lt/0/3,
        Lb/0/1,
        A/1/2,
        B/2/3,
        M/2/2,
        F/2/1,
        T/3/4,
        Cc/3/3,
        D/3/2,
        E/3/1,
        U/3/0,
        N/5/2
    }{
        \pgfmathsetmacro{\radial}
            {\OrbitR+\FishScale*(\u-2.5)}
        \pgfmathsetmacro{\tangent}
            {\FishScale*(\v-2)}

        \coordinate (\P\k) at
        ({
            \radial*cos(\ang)
            -\tangent*sin(\ang)
         },{
            \radial*sin(\ang)
            +\tangent*cos(\ang)
         });
    }

    \fill[gray!65,opacity=0.38]
        (Lt\k) -- (Lb\k) -- (A\k) --
        (F\k) -- (E\k) -- (U\k) --
        (N\k) -- (T\k) -- (Cc\k) --
        (B\k) -- (A\k) -- cycle;

    \draw[thick] (Lt\k) -- (Lb\k);
    \draw[thick] (Lt\k) -- (A\k);
    \draw[thick] (Lb\k) -- (A\k);

    \draw[thick] (A\k) -- (B\k);
    \draw[thick] (A\k) -- (F\k);
    \draw[thick] (B\k) -- (Cc\k);
    \draw[thick] (F\k) -- (E\k);

    \draw[thick] (A\k) -- (N\k);
    \draw[thick] (T\k) -- (U\k);

    \draw[thick] (T\k) -- (N\k);
    \draw[thick] (U\k) -- (N\k);
    \draw[thick] (Cc\k) -- (N\k);
    \draw[thick] (E\k) -- (N\k);

    \draw[thick] (B\k) -- (M\k);
    \draw[thick] (M\k) -- (E\k);
    \draw[thick] (M\k) -- (F\k);
    \draw[thick] (M\k) -- (Cc\k);

    \node at
    ({
        (\OrbitR+\BallR+0.38)*cos(\ang)
     },{
        (\OrbitR+\BallR+0.38)*sin(\ang)
     })
    {$P_{\idx}$};
}

\end{tikzpicture}
\caption{A same-dimensional realization of the rotational construction
for $n=2$. Each copy $P_k$ is contained in an enclosing disk that lies
strictly inside one of the five angular sectors.}
\label{fig:rotation-same-dimension}
\end{figure}

\begin{thm}\label{thm:completeness}
\quad
\begin{enumerate}
\item
$\dyn{\lgc{TLR}}$ is sound and complete for the class
$\mathfrak{hDT}$.

\item
$\dyn{\lgc{ALR}}$ is sound and complete both for the class
$\mathfrak{hDF}$ and for the class $\mathfrak{hDA}$.

\item
$\dyn{\lgc{PLR}}$ is sound and complete for the class
$\mathfrak{hDP}$. In fact, it is already complete for the subclass of
h-dynamic polyhedral spaces whose dynamics is a rotation.
\end{enumerate}
\end{thm}

\begin{proof}
Soundness was established in Theorem~\ref{thm:soundness}. For
completeness, suppose that
\[
\dyn{\Lambda}\nvdash\varphi.
\]
By Lemma~\ref{lem:g_equivalence},
\[
\Lambda\nvdash g(\varphi).
\]
Now Theorem~\ref{t:comp_no-hD} provides a model $\mathcal M=(X,V)$ in the corresponding spatial class
and a point $x\in X$ such that
\[
\mathcal M,x\models\lnot g(\varphi).
\]
Since $g$ commutes with negation, this is equivalent to
\[
\mathcal M,x\models g(\lnot\varphi).
\]

For $\Lambda=\lgc{TLR}$ or $\Lambda=\lgc{ALR}$, apply
Lemma~\ref{lemmGtoMod} to the formula $\lnot\varphi$. The resulting
h-dynamic model satisfies $\lnot\varphi$. Moreover, it is finite when
the original model is finite and Alexandroff when the original model
is Alexandroff.

For $\Lambda=\lgc{PLR}$, apply
Lemma~\ref{lem:poly_reconstruction} to the formula $\lnot\varphi$.
This yields an h-dynamic polyhedral model, whose dynamics is a
rotation, satisfying $\lnot\varphi$. Thus, in every case,
$\varphi$ has a countermodel in the required class.
\end{proof}

\section{Conclusions and Further Work}

We have proposed three modifications of Dynamic Topological Logic: first, by noting that systems with homeomorphisms can naturally be endowed with backwards-looking temporal modalities; second, by considering semantics over polyhedra, and finally, by enriching the language with spatial reachability operators.
It is worth remarking that these modifications work independently and can be combined in various ways, and indeed we have axiomatised not only the dynamic reachability logic of polyhedra, but also of the more standard topological and Alexandroff spaces. It is worth noting that, unlike in the $\gamma$-free logic, the latter two classes lead to different dynamic topological logics.

Meanwhile, there have been many developments in the three decades since Artemov and co-authors first introduced DTL, leading to many natural open questions for the polyhedral variants.
The original logic of Artemov et al.~\cite{ArtemovDavorenNerode1997-modal-logics-and-topological-semantics-for-hybrid-systems} focuses on dynamical systems where the function is continuous, but not necessarily invertible.
We conjecture that the `past'-free fragment of our logic is sound and complete for this class, although there are some technical challenges when dealing with $\gamma$ in this context.

\begin{question}
Is the back-free fragment of dynamic polyhedral reachability logic complete for the class of polyhedra with a PL endofunction?
\end{question}

A large portion of the literature has also focused on logics with `eventually'. In this context, DTL is undecidable~\cite{KonevKontchakovWolterZakharyaschev2006-dtl-over-spaces-with-continuous-functions} and non-finitely axiomatisable~\cite{Fernandez2012-DTL-non-finite-axiomatizability} (or even non-axiomatisable when $f$ is a homeomorphism~\cite{KonevKontchakovWolterZakharyaschev2006-on-dynamic-topological-and-metric-logics}), but it is computably enumerable~\cite{Fernandez2009-non-det-quasi-models} and enjoys an axiomatisation in an extended language~\cite{Fernandez2012-DTLstar-axiomotization}.
Particularly promising is the fact that DTL over a $\sf GL$ base {\em is} finitely axiomatisable~\cite{Untangled}, which will most likely carry over to $\sf Grz$ and hence to polyhedral logics.
Recalling that $\sf Grz$ defines the class of hereditarily irresolvable spaces \cite{BEZHANISHVILI2003291}, but is also complete for the subclass of scattered spaces, we arrive at:

\begin{question}
Is $\sf DTL$ over a $\sf Grz$ base complete for the class of dynamic scattered spaces, or at least for the class of dynamic hereditarily irresolvable spaces?
\end{question}

However, `eventually' is defined as an infinite union, which does not produce polyhedral sets in general.
This produces a theoretical challenge in finding the `right' DTL of polyhedra, in the presence of infinitary tenses.
This could be remedied, for example, by restricting to fixed triangulations, corresponding to finite $\sf Grz$ frames, but currently no version of DTL with the finite model property seems to be known in the presence of `eventually', except in the case of $\sf S5$ topologies~\cite{Kremer09}.
This immediately raises the following question.

\begin{question}
Is DTL of finite topological spaces decidable?
What about DTL of finite $\sf Grz$ spaces?
\end{question}

One promising exception here is DTL with finite iterations, which is decidable and enjoys the FMP~\cite{GabelaiaKuruczWolterZakharyaschev2006-expanding-domain-producs}.
However, such logics have only been treated semantically, leading to another natural open question.

\begin{question}
    Is DTL with finite iterations finitely axiomatisable, either over the class of all topological spaces, or over the class of all $\sf Grz$ spaces?
\end{question}

An affirmative answer to this question should be a major step in settling the following.,

\begin{question}
    Is dynamic polyhedral logic with `eventually' finitely axiomatisable over the class of dynamic polyhedra with finite iterations, with or without $\gamma$?
\end{question}

\bibliographystyle{spmpsci}
\bibliography{main}

@mastersthesis{Kopnev,
    author = {Kopnev, Kirill},
    title = {Dynamic logics of polyhedra and their application in 3D modeling},
    school = {Intitute for Logic, Language and Computation, University of Amsterdam},
    year = {2023}
}

@article{BezhanishviliMarraMcNeillPedrini2018,
  author  = {Bezhanishvili, Nick and Marra, Vincenzo and McNeill, Daniel and Pedrini, Andrea},
  title   = {Tarski's Theorem on Intuitionistic Logic, for Polyhedra},
  journal = {Annals of Pure and Applied Logic},
  volume  = {169},
  number  = {5},
  pages   = {373--391},
  year    = {2018},
  doi     = {10.1016/j.apal.2017.12.005}
}

@article{AdamDayBezhanishviliGabelaiaMarra2024,
  author  = {Adam-Day, Sam and Bezhanishvili, Nick and Gabelaia, David and Marra, Vincenzo},
  title   = {Polyhedral Completeness of Intermediate Logics: The Nerve Criterion},
  journal = {The Journal of Symbolic Logic},
  volume  = {89},
  number  = {1},
  pages   = {342--382},
  year    = {2024},
  doi     = {10.1017/jsl.2022.76}
}

@article{BEZHANISHVILI2003291,
title = {Scattered, Hausdorff-reducible, and hereditarily irresolvable spaces},
journal = {Topology and its Applications},
volume = {132},
number = {3},
pages = {291-306},
year = {2003},
issn = {0166-8641},
doi = {https://doi.org/10.1016/S0166-8641(03)00039-7},
url = {https://www.sciencedirect.com/science/article/pii/S0166864103000397},
author = {Guram Bezhanishvili and Ray Mines and Patrick J. Morandi}
}

@misc{gabelaia2018modallogicplanarpolygons,
      title={Modal logic of planar polygons}, 
      author={David Gabelaia and Kristina Gogoladze and Mamuka Jibladze and Evgeny Kuznetsov and Maarten Marx},
      year={2018},
      eprint={1807.02868},
      archivePrefix={arXiv},
      primaryClass={math.LO},
      url={https://arxiv.org/abs/1807.02868}, 
}

@book{RourkeSanderson1972,
  author    = {Rourke, C. P. and Sanderson, B. J.},
  title     = {Introduction to Piecewise-Linear Topology},
  publisher = {Springer-Verlag},
  address   = {Berlin, Heidelberg, New York},
  year      = {1972}
}

@inproceedings{gagarin,
    author = {Aleksandr Gagarin and David Fern\'andez-Duque},
    title = {Topological Logics of Path-reachability},
    booktitle = {Advances in Modal Logic},
    year = {2026}
}

@book{G92,
	Author = {R. Goldblatt},
	Edition = 2,
	Number = 7,
	Series = {{CSLI Lecture Notes}},
	Publisher = {Center for the Study of Language and Information},
	Title = {{{L}ogics of {T}ime and {C}omputation}},
	note = {second edition},
	Year = 1992
}

@article{Kremer09,
  author       = {Philip Kremer},
  title        = {Dynamic topological {S5}},
  journal      = {Ann. Pure Appl. Log.},
  volume       = {160},
  number       = {1},
  pages        = {96--116},
  year         = {2009},
  url          = {https://doi.org/10.1016/j.apal.2009.01.015},
  doi          = {10.1016/J.APAL.2009.01.015},
  bibsource    = {dblp computer science bibliography, https://dblp.org}
}

@INPROCEEDINGS{robotics,
  author={Ijspeert, A.J. and Nakanishi, J. and Schaal, S.},
  booktitle={Proceedings 2002 IEEE International Conference on Robotics and Automation (Cat. No.02CH37292)},
  title={Movement imitation with nonlinear dynamical systems in humanoid robots},
  year={2002},
  volume={2},
  number={},
  pages={1398-1403 vol.2},
  doi={10.1109/ROBOT.2002.1014739}}

@inproceedings{Untangled,
  author       = {David Fern{\'{a}}ndez{-}Duque and
                  Yo{\`{a}}v Montacute},
  editor       = {Brian Williams and
                  Yiling Chen and
                  Jennifer Neville},
  title        = {Untangled: {A} Complete Dynamic Topological Logic},
  booktitle    = {Thirty-Seventh {AAAI} Conference on Artificial Intelligence, {AAAI}
                  2023, Thirty-Fifth Conference on Innovative Applications of Artificial
                  Intelligence, {IAAI} 2023, Thirteenth Symposium on Educational Advances
                  in Artificial Intelligence, {EAAI} 2023, Washington, DC, USA, February
                  7-14, 2023},
  pages        = {6355--6362},
  publisher    = {{AAAI} Press},
  year         = {2023},
  url          = {https://ojs.aaai.org/index.php/AAAI/article/view/25782},
  bibsource    = {dblp computer science bibliography, https://dblp.org}
}

@ARTICLE{selfdrive,
  author={Dickmanns, Ernst D. and Mysliwetz, B. D. and Christians, T.},
  journal={IEEE Transactions on Systems, Man, and Cybernetics},
  title={An integrated spatio-temporal approach to automatic visual guidance of autonomous vehicles},
  year={1990},
  volume={20},
  number={6},
  pages={1273-1284},
  doi={10.1109/21.61200}}

@book{akin,
author={E. Akin},
title={The General Topology of Dynamical Systems},
series={Graduate Studies in Mathematics},
publisher={American Mathematical Society},
year={1993},
}

@article{selfdrive2,
	abstractnote = {In order to be trusted by humans, Artificial Intelligence agents should be able to describe rationales behind their decisions. One such application is human action recognition in critical or sensitive scenarios, where trustworthy and explainable action recognizers are expected. For example, reliable pedestrian action recognition is essential for self-driving cars and explanations for real-time decision making are critical for investigations if an accident happens. In this regard, learning-based approaches, despite their popularity and accuracy, are disadvantageous due to their limited interpretability. This paper presents a novel neuro-symbolic approach that recognizes actions from videos with human-understandable explanations. Specifically, we first propose to represent videos symbolically by qualitative spatial relations between objects called qualitative spatial object relation chains. We further develop a neural saliency estimator to capture the correlation between such object relation chains and the occurrence of actions. Given an unseen video, this neural saliency estimator is able to tell which object relation chains are more important for the action recognized. We evaluate our approach on two real-life video datasets, with respect to recognition accuracy and the quality of generated action explanations. Experiments show that our approach achieves superior performance on both aspects to previous symbolic approaches, thus facilitating trustworthy intelligent decision making. Our approach can be used to augment state-of-the-art learning approaches with explainabilities.},
	author = {Hua, Hua and Li, Dongxu and Li, Ruiqi and Zhang, Peng and Renz, Jochen and Cohn, Anthony},
	doi = {10.1609/aaai.v36i5.20513},
	journal = {Proceedings of the AAAI Conference on Artificial Intelligence},
	month = {Jun.},
	number = {5},
	pages = {5710-5718},
	title = {Towards Explainable Action Recognition by Salient Qualitative Spatial Object Relation Chains},
	url = {https://ojs.aaai.org/index.php/AAAI/article/view/20513},
	volume = {36},
	year = {2022}}

@inproceedings{PolyCompleteness,
    author = {Bezhanishvili, N. and Bussi, L. and Ciancia, V. and Fernandez-Duque, D. and Gabelaia, D.},
    title = {Logics of Polyhedral Reachability},
    booktitle = {Advances in Modal Logic},
    year = {2024},
    publisher = {College Publications},
    pages = {187-204},
    isbn = {978-1-84890-467-5}
}

@article{CLLM16,
  title    = {{Model Checking Spatial Logics for Closure Spaces}},
  author   = {Ciancia, V. and Latella, D. and Loreti, M. and Massink, M.},
  doi      = {10.2168/LMCS-12(4:2)2016},
  journal  = {{Logical Methods in Computer Science}},
  volume   = {Volume 12, Issue 4},
  year     = {2016},
  month    = Oct
}

@inproceedings{BCLM19,
  author    = {Belmonte, G. and Ciancia, V. and Latella, D. and Massink, M.},
  title     = {VoxLogicA: {A} Spatial Model Checker for Declarative Image Analysis},
  booktitle = {Tools and Algorithms for the Construction and Analysis of Systems, {TACAS}},
  series    = {LNCS},
  volume    = "11427",
  pages     = {281--298},
  publisher = {Springer},
  year      = {2019},
  doi       = {10.1007/978-3-030-17462-0\_16}
}

@article{BCGGLM22,
  title    = {Geometric Model Checking of Continuous Space},
  author   = {Bezhanishvili, N. and Ciancia, V. and Gabelaia, D. and Grilletti, G. and Latella, D. and Massink, M.},
  doi      = {10.46298/lmcs-18(4:7)2022},
  journal  = {Logical Methods in Computer Science},
  volume   = {Volume 18, Issue 4},
  year     = {2022},
  month    = {Nov},
}

@inproceedings{BCGLM22,
   author    = {Bussi, L.
                and Ciancia, V.
                and Gadducci, F.
                and Latella, D.
                and Massink, M.},
   editor    = {Bowles, J.
                and Broccia, G.
                and Pellungrini, R.},
   title     = {Towards Model Checking Video Streams Using
{V}ox{L}ogic{A} on {GPU}s},
   booktitle = {From Data to Models and Back},
   series    = {LNCS},
   volume    = {13268},
   year      = {2022},
   publisher = {Springer},
   pages     = {78--90},
   doi       = {10.1007/978-3-031-16011-0_6}
}

@article{BBCJLMV26,
    title      = {Weak Simplicial Bisimilarity and Minimisation for Polyhedral Model Checking},
    author     = {Nick Bezhanishvili and Laura Bussi and Vincenzo Ciancia and David Gabelaia and Mamuka Jibladze and Diego Latella and Mieke Massink and Erik P. de Vink},
    url        = {https://lmcs.episciences.org/14827},
    doi        = {10.46298/lmcs-22(1:6)2026},
    journal    = {Logical Methods in Computer Science},
    issn       = {1860-5974},
    volume     = {Volume 22, Issue 1},
    eid        = 6,
    year       = {2026},
    month      = {Jan},
}

@InProceedings{FASE26,
author="Belmonte, Gina
and Ciancia, Vincenzo
and Latella, Diego
and Massink, Mieke",
editor="Albert, Elvira
and Pasareanu, Corina",
title="Model Checking in Space with Applications to Medical Image Analysis",
booktitle="Fundamental Approaches to Software Engineering",
year="2026",
publisher="Springer Nature Switzerland",
address="Cham",
pages="3--18",
isbn="978-3-032-22774-4"
}

@article{BCM25,
title = {Symbolic and hybrid AI for brain tissue segmentation using spatial model checking},
journal = {Artificial Intelligence in Medicine},
volume = {167},
pages = {103154},
year = {2025},
issn = {0933-3657},
doi = {https://doi.org/10.1016/j.artmed.2025.103154},
url = {https://www.sciencedirect.com/science/article/pii/S0933365725000892},
author = {Gina Belmonte and Vincenzo Ciancia and Mieke Massink}
}

@ARTICLE{McKinseyTarski1944-the-algebra-of-topology,
  AUTHOR = {McKinsey, John Charles Chenoweth and Tarski, Alfred},
  DOI = {10.2307/1969080},
  EPRINT = {1969080},
  EPRINTTYPE = {jstor},
  JOURNAL = {Annals of Mathematics},
  NUMBER = {1},
  PAGES = {141--191},
  TITLE = {The algebra of topology},
  VOLUME = {45},
  YEAR = {1944},
}

@TECHREPORT{ArtemovDavorenNerode1997-modal-logics-and-topological-semantics-for-hybrid-systems,
  AUTHOR = {Artemov, Sergei Nikolaevich and Davoren, Jennifer M. and Nerode, Anil},
  INSTITUTION = {Cornell University},
  NUMBER = {MSI 97-05},
  TITLE = {Modal Logics and Topological Semantics for Hybrid Systems},
  YEAR = {1997},
}

@ARTICLE{KremerMints2005-DTL,
  AUTHOR = {Kremer, Philip and Mints, Grigori},
  PUBLISHER = {Elsevier},
  DOI = {10.1016/j.apal.2004.06.004},
  ISSN = {0168-0072},
  JOURNAL = {Annals of Pure and Applied Logic},
  NUMBER = {1},
  PAGES = {133--158},
  TITLE = {Dynamic topological logic},
  VOLUME = {131},
  YEAR = {2005},
}

@ARTICLE{GabelaiaKuruczWolterZakharyaschev2006-expanding-domain-producs,
  AUTHOR = {Gabelaia, David and Kurucz, Agi and Wolter, Frank and Zakharyaschev, Michael},
  DOI = {10.1016/j.apal.2006.01.001},
  ISSN = {0168-0072},
  JOURNAL = {Annals of Pure and Applied Logic},
  NUMBER = {1-3},
  PAGES = {245--268},
  TITLE = {Non-primitive recursive decidability of products of modal logics with expanding domains},
  VOLUME = {142},
  YEAR = {2006},
}

@INPROCEEDINGS{KonevKontchakovWolterZakharyaschev2006-dtl-over-spaces-with-continuous-functions,
  AUTHOR = {Konev, Boris and Kontchakov, Roman and Wolter, Frank and Zakharyaschev, Michael},
  EDITOR = {Governatori, Guido and Hodkinson, Ian and Venema, Yde},
  LOCATION = {London},
  PUBLISHER = {College Publications},
  BOOKTITLE = {Advances in Modal Logic},
  EVENTDATE = {2006-09},
  ISBN = {1-904987-20-6},
  PAGES = {299--318},
  TITLE = {Dynamic Topological Logics Over Spaces with Continuous Functions},
  VENUE = {Noosa, Queensland, Australia},
  VOLUME = {6},
  YEAR = {2006},
}

@ARTICLE{KonevKontchakovWolterZakharyaschev2006-on-dynamic-topological-and-metric-logics,
  AUTHOR = {Konev, Boris and Kontchakov, Roman and Wolter, Frank and Zakharyaschev, Michael},
  DOI = {10.1007/s11225-006-9005-x},
  JOURNAL = {Studia Logica},
  NUMBER = {1},
  PAGES = {129--160},
  TITLE = {On Dynamic Topological and Metric Logics},
  VOLUME = {84},
  YEAR = {2006},
}

@ARTICLE{Fernandez2009-non-det-quasi-models,
  AUTHOR = {Fern\'{a}ndez-Duque, David},
  DOI = {10.1016/j.apal.2008.09.015},
  ISSN = {0168-0072},
  JOURNAL = {Annals of Pure and Applied Logic},
  NOTE = {Kurt G\"{o}del Centenary Research Prize Fellowships},
  NUMBER = {2},
  PAGES = {110--121},
  TITLE = {Non-deterministic semantics for dynamic topological logic},
  VOLUME = {157},
  YEAR = {2009},
}

@ARTICLE{Fernandez2012-DTLstar-axiomotization,
  AUTHOR = {Fern\'{a}ndez-Duque, David},
  DOI = {10.2178/jsl/1344862169},
  JOURNAL = {The Journal of Symbolic Logic},
  NUMBER = {3},
  PAGES = {947--969},
  TITLE = {A sound and complete axiomatization for Dynamic Topological Logic},
  VOLUME = {77},
  YEAR = {2012},
}

@INPROCEEDINGS{Fernandez2012-DTL-non-finite-axiomatizability,
  AUTHOR = {Fern\'{a}ndez-Duque, David},
  EDITOR = {Bolander, Thomas and Bra\"{u}ner, Torben and Ghilardi, Silvio and Moss, Lawrence},
  LOCATION = {London},
  PUBLISHER = {College Publications},
  BOOKTITLE = {Advances in Modal Logic},
  EVENTDATE = {2012-08},
  ISBN = {978-1-84890-068-4},
  PAGES = {200--216},
  TITLE = {Non-finite Axiomatizability of Dynamic Topological Logic},
  VENUE = {Copenhagen, Denmark},
  VOLUME = {9},
  YEAR = {2012},
}

\end{document}